\documentclass[preprint]{sig-alternate}
\usepackage{amsmath}
\usepackage{url}
\usepackage{amssymb}
\usepackage{amsmath}
\usepackage{algorithm,subfigure}
\usepackage{graphicx}
\usepackage[noend]{algorithmic}

\CopyrightYear{2013}
\toappear{Permission to make digital or hard copies of all or part of this work for
personal or classroom use is granted without fee provided that copies are
not made or distributed for profit or commercial advantage and that copies
bear this notice and the full citation on the first page. To copy otherwise, to
republish, to post on servers or to redistribute to lists, requires prior specific
permission and/or a fee. Copyright 2013
}

\newtheorem{lemma}{Lemma}
\newtheorem{theorem}{Theorem}
\newtheorem{definition}{Definition}
\begin{document}

\title{Maximizing Acceptance Probability for Active Friending in On-Line
Social Networks}
\author{De-Nian Yang$^\dagger$,\ \ Hui-Ju Hung$^\dagger$, \ \ Wang-Chien Lee$%
^\star$,\ \ Wei Chen$^\mathsection$ \\
\affaddr{$^\dagger$Academia Sinica, Taipei, Taiwan} \\
\affaddr{$^\star$The Pennsylvania State University, State College,
Pennsylvania, USA} \\
\affaddr{$^\mathsection$Microsoft Research Asia, Beijing, China} \\
}
\maketitle

\begin{abstract}
Friending recommendation has successfully contributed to the explosive
growth of on-line social networks. Most friending recommendation services
today aim to support passive friending, where a user passively selects
friending targets from the recommended candidates. In this paper, we
advocate recommendation support for active friending, where a user actively
specifies a friending target. To the best of our knowledge, a recommendation
designed to provide guidance for a user to systematically approach his
friending target, has not been explored in existing on-line social
networking services. To maximize the probability that the friending target
would accept an invitation from the user, we formulate a new optimization
problem, namely, \emph{Acceptance Probability Maximization (APM)}, and
develop a polynomial time algorithm, called \emph{Selective Invitation with
Tree and In-Node Aggregation (SITINA)}, to find the optimal solution. We
implement an active friending service with SITINA in Facebook to validate
our idea. Our user study and experimental results manifest that SITINA
outperforms manual selection and the baseline approach in solution quality
efficiently.
\end{abstract}

\keywords{Friending, social network, social influence}

\numberofauthors{1} \numberofauthors{1}

\section{Introduction}

\label{sec:introduction}

Due to the development and popularity of social networking services, such as
Facebook, Google+, and LinkedIn, the new notion of \textquotedblleft social
network friending\textquotedblright\ has appeared in recent years. To boost
the growth of their user bases, existing social networking services usually
provide friending recommendations to their users, encouraging them to send
invitations to make more friends.
Conventionally, friending recommendations are made following a \textit{%
passive friending} strategy, i.e., a user passively selects candidates from
the provided recommendation list to send the invitations.
Moreover, the recommended candidates are usually friends-of-friends of the
user, especially those who share many common friends with the user.
This strategy is quite intuitive because friends-of-friends may have been
acquaintances or friends offline. Furthermore, most users may feel more
comfortable to send a friending invitation to friends-of-friends rather than
a total stranger who they have shared no social connections with at all. It
is envisaged that the success rate of such a passive friending strategy is
high, contributing to the explosive growth of on-line social networking
services.

In contrast to the passive friending, the idea of \textit{active friending},
where a person may take proactive actions to make friend with another
person, does exist in our everyday life.
For example, in a high school, a student fan may like to make friend with
the captain in the school soccer team or with the lead singer in a
rock-and-roll band of the school. A salesperson may be interested in getting
acquainted with a high-value potential customer in hope of making a business
pitch. A young KDD researcher may desire to make friend with the leaders of
the community to participate in conference organizations and services.
However, to the best of our knowledge, the idea of providing friending
recommendations to assist and guide a user to effectively approaching
another person for active friending has not been explored in existing
on-line social networking services. We argue that social networking service
providers, interested in exploring new revenues and further growth of their
user bases, may be interested in supporting active friending.

One may argue that, in existing social networking services, an \emph{active
friending initiator} can send an invitation directly to the \emph{friending
target} anyway.\footnote{%
For the rest of the paper, we refer to the friending initiator and friending
target as \emph{initiator} and \emph{target} for short.} However, it may not
work if the initiator is regarded as a stranger by the target, especially
when they are socially distant, i.e., they have no common friends.
Therefore, to increase the chance that the target would accept the friending
invitation, it may be a good idea for the initiator to first know some
friends of the target, which in turn may require the initiator to know some
friends of friends of the target. In other words, if the initiator would
like to plan for some actions, he may need the topological information of
the social network between the target and himself, which unfortunately is
not available due to privacy concerns.
Therefore, it would be very nice if the social networking service providers,
given a target specified by the initiator, could provide a step-by-step
guidance in form of recommendations to assist the initiator to make friends
towards the target.

In this paper, we are making a grand suggestion for the social networking
service providers to support active friending. Our sketch is as follows. By
iteratively recommending a list of candidates who are friends of at least
one existing friend of the initiator, a social networking service provider
may support active friending,
without violating the current practice of privacy preservation in
recommendations.
Consider an initiator who specifies a friending target.
The social networking service, based on its proprietary algorithms,
recommends a set of friending candidates who may likely increase the chance
for the target to accept the eventual invitation from the initiator. Similar
to the recommendations for passive friending, the recommendation list
consists of only the friends of existing friends of the initiator.
Supposedly, the initiator follows the recommendations to send invitations to
candidates in the list. The invitation is displayed to a candidate along
with the list of common friends between the initiator and the candidate so
as to encourage acceptance of the invitation.\footnote{%
This is also a common practice for passive friending in existing social
networking services such as Facebook, Google+, and LinkedIn.}
As such, the aforementioned step is repeated until the friending target
appears in the recommendation list and an invitation is sent by the
initiator.
Obviously, the recommendations made for passive friending may not work well
because active friending is target-oriented. The recommended candidates
should be carefully chosen for the initiator, guiding him to approach the
friending target step-by-step.

To support active friending, the key issue is on the design of the
algorithms that select the recommendation candidates. A simple scheme is to
provide recommendations by unveiling the shortest path between the initiator
and the target in the social network, i.e., recommending one candidate at
each step along the path. As such, the initiator can gradually approach the
target by acquainting the individuals on the path. However, this
shortest-path recommendation approach may fail as soon as a middle-person
does not accept the friending invitation (since only one candidate is
included in the recommendation list for each step). To address this issue,
it is desirable to recommend multiple candidates at each step since the
initiator is more likely to share more common friends with the target and
thereby more likely to get accepted by the target. Especially, by \textit{%
broadcasting} the friending invitations to all neighbors of the initiator's
friends, the probability to reach the friending target and get accepted can
be effectively maximized as enormous number of paths are flooded with
invitations to approach the target. Nevertheless, friending invitations are
abused here because the above undirectional broadcast is aimless and prone
to involve many unnecessary neighbors. Moreover, the initiator may not want
to handle a large number of tedious invitations.

In this paper, we study a new optimization problem, called \textit{%
Acceptance Probability Maximization (APM)}, for active friending in on-line
social networks. The service providers, who eager to explore new monetary
tools for revenue increase, may consider to charge the users from active
friending service.\footnote{%
Recent news reported that Facebook now allows its user to pay to promote
their and their friends' posts \cite{CNN}.} Given an initiator $s$, a
friending target $t$, and the maximal number $r_{R}$ of invitations allowed
to be issued by the initiator, APM\ finds a set $R$ of $r_{R}$ nodes, such
that $s$ can sequentially send invitations to the nodes in $R$ in order to
approach $t$. The objective is to maximize the acceptance probability at $t$
of the friending invitation when $s$ send it to $t$. The parameter $r_{R}$
controls the trade-off between the expected acceptance probability of $t$
and the anticipated efforts made by $s$ for active friending $t$.\footnote{%
Since $s$\ is not aware of the network topology and the distance to $t$, it
is not reasonable to let $s$\ directly specify $r_{R}$. Instead, it is more
promising for the service provider to list a set of $r_{R}$\ and the
corresponding acceptance probabilities and monetary costs, so that the user
can choose a proper $r_{R}$\ according to her available budget.}
Again, $R$ is not returned to $s$ as a whole due to privacy concerns.
Instead, only a subset $R_{s}$ of nodes that are adjacent to the existing
friends of $s$ are recommended to $s$, while other subsets of $R$ will be
recommended to $s$ as appropriate in later steps\footnote{%
In this paper, APM is formulated as an offline optimization problem aiming
to maximize the acceptance probability in expectation. In an on-line
scenario where the initiator does not send invitations to some nodes in $%
R_{s}$ or some nodes in $R_{s}$ do not accept the invitations, a new APM
with renewed invitation budget could be re-issued to obtain adapted
recommendations. While this scenario raises important issues, it is beyond
the scope of this paper.}. 

To tackle the APM problem, we propose three algorithms: i) Range-based
Greedy (RG) algorithm, ii) Selective Invitation with Tree Aggregation (SITA)
algorithm, and iii) Selective Invitation with Tree and In-Node Aggregation
(SITINA) algorithm.
RG selects candidates by taking into account their acceptance probability
and the remaining budget of invitations, leading to the best recommendations
for each step. However, the algorithm does not achieve the optimal
acceptance probability of the invitation to a target due to the lack of
coordinated friending efforts. On the other hand, aiming to systematically
select the nodes for recommendation, 
SITA is designed by dynamic programming to find nodes which may result in a
coordinated friending effort to increase the acceptance probability of the
target. SITA is able to obtain the optimal solution, yet has an exponential
time complexity.
To address the efficiency issue, SITINA further refines the ideas in SITA by
carefully aggregating some information gathered during processing to
alleviate redundant computation in future steps and
thus obtains the optimal solution for APM in polynomial time. The
contributions of this paper are summarized as follows.


\begin{itemize}
\item We advocate for the idea of \emph{active friending} in on-line social
networks and propose to support active friending through a series of
recommendation lists which serve as a step-to-step guidance for the
initiator. 

\item We formulate a new optimization problem, namely, Acceptance
Probability Maximization (APM), for configuring the recommendation lists in
the active friending process. APM aims to maximize the acceptance
probability of the invitation from the initiator to the friending target, by
recommending selective intermediate friends to approach the target. 

\item We propose a number of new algorithms for APM. Among them, Selective
Invitation with Tree and In-Node Aggregation (SITINA) derives the optimal
solution for APM with $O(n_{V}{r_{R}}^{2})$ time, where $n_{V}$ is the
number of nodes in a social network, and ${r_{R}}$ is the number of
invitations budgeted for APM.

\item We implement SITINA in Facebook in support of active friending and
conduct a user study including 169 volunteers with varied background. The
user study and experimental results manifest that SITINA outperforms manual
selection and the baseline approach in solution quality efficiently.
\end{itemize}

The rest of this paper is organized as follows. Section \ref%
{sec:problem_formulation} introduces a model for invitation acceptance and
formulates APM. Section \ref{sec:related_work} reviews the related work.
Section \ref{sec:algorithm} presents the SITA\ and SITINA algorithms
proposed for APM. Section \ref{sec:experiments} reports our user study and
experimental results. Finally, Section~\ref{sec:conclusion} concludes the
paper.


\section{Invitation Acceptance}

\label{sec:problem_formulation}

The notion of \emph{acceptance probability} is with respect to an
invitation. Thus, here we first discuss two important factors that may
affect the acceptance probability of a friending invitation in the
environment of on-line social networking services and describe how in this
work we determine whether an individual would accept a received invitation.
Next, we explain why the issue of deriving the acceptance probability over a
social network is very challenging and how we address this issue by adopting
an approximate probability based on a maximum influence in-arborescence
(MIIA) tree. We formulate the acceptance probability maximization (APM)
problem based on the MIIA tree. The invitation acceptance model follows the
existing social influence and homophily models, which have been justified in
the literature. Later in Section \ref{sec:experiments}, the invitation
acceptance model will be validated by a user study with 169 volunteers.

\subsection{Factors for Invitation Acceptance}

In the process of active friending, while friending candidates are
recommended for the initiator to send invitations, whether the invitees will
accept the invitations remains uncertain. Based on prior research in
sociology and on-line social networks~\cite{Fond10WWW, Marsden93SMR,
McPherson01ARS}, we argue that when a person receives an invitation over an
on-line social network, the decision of the invitee depends primarily on two
important factors: i) the \textit{social influence factor }\cite%
{Goyal10WSDM,Goyal11VLDB},
and ii) the \textit{homophily factor }\cite%
{Cialdini04ARP,Fond10WWW,McPherson01ARS}.
Here, the social influence factor represents the influence from the
surroundings (i.e., common friends) of individuals in the social network on
the decision. On the other hand, the homophily factor captures the fact that
each individual in a social network has a distinctive set of personal
characteristics, and the similarities and compatibilities among
characteristics of two individuals can strongly influence whether they will
become friends~\cite{Fond10WWW}. Between them, social influence comes from
established social links, while the homophily between two individuals may
exist without a prerequisite of established social relationship. Thus, we
consider these two factors separately but aim to treat them in a uniformed
fashion in our derivation of the acceptance probability for an invitation.

As the social influence factor involves the structure of social network
(i.e., the common friends of the individuals), we first consider the
acceptance probability of an invitation in terms of social influence.%
\footnote{%
We intend to extend it with homophily factor later.} Let the social network
be represented as a social graph $G(V,E)$ where $V$ consists of all the
users in the social networking system and $E$ be the established social
links among the users. An edge weight $w_{u,v}\in \lbrack 0,1]$ on the
directed edge $(u,v)\in E$ probabilistically denotes the social influence of
$u$ upon $v$. The probability can be derived according to the existing
method \cite{Goyal10WSDM,Goyal11VLDB} according to the interaction in
on-line social networks, while the setting of negative social influence has
also been introduced in \cite{Chen11SDM}. Thus, if $u$ is associated with an
invitation from a user $s$ to $v$ (i.e., $u$ is a common friend of $s$ and $%
v $), $w_{u,v}$ is the probability for $v$ to be socially influenced by $u$
to accept the invitation.\footnote{%
The social influence probability has been extensively used
to quantify the probability of success in the process of conformity,
assimilation, and persuasion in Social Psychology~\cite{Cialdini04ARP,
Fond10WWW, Marsden93SMR}. While how to obtain the edge weight is an active
research topic~\cite{Goyal10WSDM,Saito08KES}, it is out of scope of this
paper.} Hence, the acceptance probability for an invitation can be derived
by taking into account the social influences of all the existing common
friends associated with an invitation. It is assumed that each common friend
$u$ has an independent social influence on the invitee $v$ to accept the
friending invitation \cite{Cialdini04ARP, Fond10WWW,Marsden93SMR}
and thus the overall acceptance probability can be obtained by aggregating
the individual social influences. Later, user study in Section \ref%
{sec:experiments} demonstrates that the influence probability and homophily
probability derived according to the literature are consistent to the real
probabilities measured from the users.


While obtaining the acceptance probability for a given invitation (as
described above) is simple, deriving the acceptance probability for a
friending target $t$ who does not have any common friend with the initiator $%
s$ becomes very challenging because more than one invitations need to be
issued (so as to make some common friends first), and there are complicated
correlations among user acceptance events for users between $s$ and $t$.

Moreover, our ultimate task is to find a set $R$ of intermediate users
between $s$ and $t$ with size at most $r_{R}$ for $s$ to send invitations
to, so as to maximize the acceptance probability of $t$. We call this
problem the acceptance probability maximization (APM) problem. Due to the
combinatorial nature of this invitation set $R$, it is still hard to find
such a set to maximize the acceptance probability of $t$ even in cases where
computing the acceptance probability is easy. The following theorem makes the
above two hardness precise.
\begin{theorem}
\label{thm:sharp_P}Given the set of neighbors S of the initiator, computing
the acceptance probability of $t$ is \#P-hard. Moreover, finding a set $R$
with size $r_{R}$ that maximizes the acceptance probability of target $t$ is
NP-hard, even for cases when computing acceptance probability is easy.
\end{theorem}
\begin{proof}

We first prove that computing the acceptance probability of $t$ with given $R$ is \#P-hard. Let $G_{R}$ denote the induced subgraph of $G$ with $s$, $t$, and $R$. Let $\overline{G}_{R}$ denote a directed subgraph of $G_{R}$ by removing every edge $(u,v)$ in $G_{R}$ if $u$ does not influence $v$ to accept an invitation, either because $u$ does not become a friend of $s$ or $v$ does not accept the invitation due to the social influence from $u$. Therefore, $t$ will finally be a friend of $s$ if there exists a path in $\overline{G}_{R}$, representing that every node in the path, including $t$, accepts the friend invitation from $s$. Apparently, if the probability of social influence associated with each edge is $0.5$, the probability that $t$ accepts the friend invitation is the number of subgraph $\overline{G}_{R}$ with $t$ accepting, divided by the number of possible subgraph $\overline{G}_{R}$, which is $2^{n_{E}}$, where $n_{E}$ is the number of directed edges. In other words, after acquiring the accepting probability, the number of subgraph $\overline{G}_{R}$ with $t$ accepting can be computed immediately by multiplying $2^{n_{E}}$.

We prove that computing the acceptance probability of $t$ is \#P-hard with the reduction from a \#P-complete problem, called $s$-$t$ \textit{connectedness problem} \cite{Valiant79SIAM}, that finds the number of subgraphs in a directed graph $G_{C}$ with a directed path from $s$ to $t$. We let $G_{R}=G_{C}$ and assign the probability of social influence with each edge as $0.5$. With the observation in the previous paragraph, if finding the accepting probability of $t$ is not \#P-complete, $s$-$t$ \textit{connectedness problem} is not \#P-complete because the number of subgraphs in $G_{C}$ with a directed path from $s$ to $t$ is simply the accepting probability of $t$ multiplied by $2^{n_{E}}$.

Moveover, even for cases when computing the acceptance probability of $t$ is easy, finding a set $R$ that maximize the acceptance probability of target $t$ is NP-hard in IC model. We prove it with a reduction from the \textit{set cover problem}. For a bipartite graph $(X, Y, E)$, set cover problem aims to identify whether there exists a $k$-node subset $X_S$ of $X$ covering all nodes in $Y$, i.e., for any $y \in Y$, there exists an $x \in X_S$ with $(x,y) \in E$. Let us denote $|Y|=z_y$. For an instance of set cover problem, we build an instance for computing the acceptance probability of $t$ as follows, and an illustration figure is shown in Figure~\ref{fig:np_reduction}. 1) We add a node $s$ and a directed edge $(s,x)$ for each $x \in X$ with weight $w(s,x)=1$. Notice that $s$ is the only one node with acceptance probability $1$ in the beginning. 2) We add a node $t$ and a directed edge $(y,t)$ for each $y \in Y$ with weight $w(y,t)=p$, $p \in (0,1)$. 3) We set the $w(e)=1$ for each $e \in E$. We prove that there is a $k$-node subset $X_S \subseteq X$ covering all nodes in $Y$ in the set cover problem if and only if there is a solution with acceptance probability $1-(1-p)^{z_y}$ when selecting $r_R=k+z_y+1$ nodes in computing the acceptance probability.\footnote{Notice that $1-(1-p)^{z_y}$ is the maximum probability when including all nodes in $X \cup Y \cup {t}$ into $R$, thus it is obvious the maximum probability when selecting $r_R$ nodes.} We first prove the sufficient condition. If there exists a $k$-node subset $X_S$ covering all nodes in $Y$, selecting $X_S \cup Y \cup \{t\}$ (totally $k+z_y+1$ nodes) will obtain acceptance probability $1-(1-p)^{z_y}$. We then prove the necessary condition. If there is a solution $R$ with $r_R$ nodes obtaining acceptance probability $1-(1-p)^{z_y}$, $R$ must contain $t$ and all nodes in $Y$, and the set $R \cap X$ (totally $r_R-z_y-1 = k$ nodes) must cover all nodes in $Y$. Thus selecting a suitable $R$ is NP-hard. The theorem follows.
\end{proof}

\begin{figure}[t]
\center
\includegraphics[width = 1.8 in]{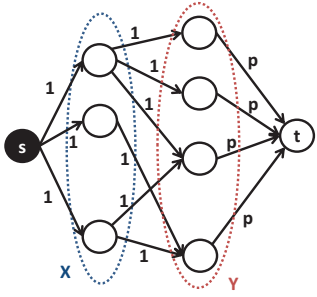}
\caption{An illustration graph of building the APM instance}
\label{fig:np_reduction}
\vspace{-9pt}
\end{figure}

\subsection{Approximate Acceptance Probability}

The spread maximization problem in Independent Cascade (IC) model \cite%
{Kempe03KDD} also faces the challenge in Theorem \ref{thm:sharp_P}.\footnote{%
The spread maximization problem, which also adopts a probabilistic influence
model, is different from APM in this paper. Given an initiator $s$ and his
friends, APM intends to discover an effective subgraph (i.e., $R$) between
the seeds and $t$. On the other hand, the spread maximization problem, given
the topology of the whole social network, aims to find a given number of
seeds to maximize the size of the whole spread $t$.} To efficiently address
this issue, an approximate IC model, called MIA \cite{Chen11SDM, Chen12AAAI,
Chen10KDD}, has been proposed. The social influence from a person $u$ to
another person $v$ is effectively approximated by their \textit{maximum
influence path} (MIP), where the social influence $w_{u,v}$ on the path $%
(u,v)$ is the maximum weight among all the possible paths from $u$ to $v$.
MIA creates a maximum influence in-arborescence, i.e., a directed tree, $%
MIIA(t,\theta )$ including the union of every MIP to $t$ with the
probability of social influence at least $\theta $ from a set $S$ of leaf
nodes. The MIA model has been widely adopted to describe the social
influence in the literature \cite{Chen11SDM, Chen12AAAI, Chen10KDD} with the
following definition on activation probability, which basically is the same
as the acceptance probability if $s$ broadcasts friending invitations to all
nodes in $MIIA(t,\theta )$. 

\begin{definition}
The activation probability of a node $v$ in $MIIA(t,\theta )$ is $ap^{\prime
}(v,S,MIIA(t,\theta )))=\newline
\left\{
\begin{matrix}
1\text{, if }v\in S \\
0\text{, if }N^{in}(v)=\emptyset \\
1-\prod_{u\in N^{in}(v)}{(1-ap^{\prime }(u,S,MIIA(t,\theta ))\cdot w_{u,v})}%
\text{, otherwise}%
\end{matrix}%
\right. $\newline
\noindent where $N^{in}(v)$ is the set of in-neighbors of $v$.
\end{definition}

Note that ${ap^{\prime }(u,S,MIIA(t,\theta ))\cdot w_{u,v}}$ is the joint
probability that $u$ is activated and successfully influences $v$, and $u$
can never influences $v$ if it is not activated. Therefore, the activation
probability of a node $v$ can be derived according to the activation
probability of all its in-neighbors, i.e., child nodes in the tree. Since $S$
is the set the leaf nodes, the activation probabilities of all nodes in ${%
MIIA(t,\theta )}$ can be derived in a bottom-up manner from $S$ toward $t$
efficiently.

In light of the similarity between the IC model and the decision model for
invitation acceptance in active friending with no budget limitation of
invitations, we also exploit MIA to tackle the APM problem. $MIIA(t,\theta )$
is constructed by the MIPs from all friends of $s$ to $t$, i.e., $S$ is the
set of friends of $s$. In other words, $\theta $ is set as $0$ to ensure
that the social influence from every friend is fully incorporated.
Nevertheless, different from the activation probability in the literature,
which allows the influence to propagate via every node in $MIIA(t,\theta )$,
the \textit{acceptance probability} for active friending allows the social
influence to take effect on invitation acceptance only via a set $R$ of
nodes to be selected in our problem. Thus, we define the acceptance
probability for an invitation to node $v$ as follows. 

\begin{definition}
The acceptance probability for an invitation of a node $v$ in $MIIA(t,\theta
)$ is $ap(v,S,R,MIIA(t,\theta )))=\newline
\left\{
\begin{matrix}
1\text{, if }v\in S \\
0\text{, if }v\notin R\text{ or }N_{v}^{in}=\emptyset \\
1-\prod_{u\in N^{in}(v),u\in R}{(1-ap(u,S,R,MIIA(t,\theta ))\cdot w_{u,v})}
\\
\text{, otherwise}%
\end{matrix}%
\right. $ \newline
\noindent where $N^{in}(s)$ is the set of in-neighbors of $s$. \label%
{def:acceptance_probability}
\end{definition}


Equipped with MIA, we are able to derive the acceptance probability of $t$\
efficiently with a simple iterative approach from the leaf nodes to the root
(i.e., $t$). The above MIA arborescence incorporates only the social
influence factor. 
As discussed earlier, the homophily factor between the initiator and the
receiver of an invitation is also crucial for friending. Homophily in \cite%
{Cialdini04ARP,Fond10WWW,McPherson01ARS} represents the probability for two
individuals $u$ and $v$ to create a new social link due to shared common
personal characteristics. Homophily in Sociology manifests the general
tendency of people to associate with others and similar others can be
quantified with varied approaches \cite%
{ChaojiNnodeLink12,HangalNodeLink10,XuNodeLink11}. The homophily probability
can be set according to \cite{Chen09CHI}.

To extend MIA, we attach a duplicated $s$ to each node with a directed edge,
with a parameter specifying the homophily factor from $s$ to $v$.
The MIP from each candidate to $t$, together with the directed edge from $s$
to the candidate, is incorporated in the extended MIA. Therefore, the
extended MIA is also an arborescence, where each leaf node is a friend of $s$
or $s$ herself, and those leaf nodes make up the set $S$.

\begin{figure}[t]
\includegraphics[width=3.2
in]{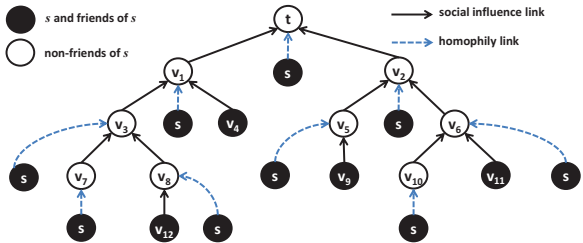} \vspace{-15pt}
\caption{Combining the social influence and homophily factors}
\label{fig:bothdimension}
\end{figure}

Figure \ref{fig:bothdimension} shows an example of the extended MIA. For
each internal node, such as $v_{1}$, its acceptance probability factors is
not only the social influence from $v_{3}$ and $v_{4}$ but also the
homophily factor between $s$ and $v_{1}$. 

In this paper, the influence probability and homophily probability are
derived according to the above literature without associating them with
different weights. Later, user study will be presented in Section \ref%
{sec:experiments}, and the results show that the real acceptance probability
complies with the acceptance probability of the above model.

\subsection{Problem Formulation}

In this work, we formulate an optimization problem, called \textit{%
Acceptance Probability Maximization (APM)}, to select a given number of
intermediate people to systematically approach the friending target $t$
based on $MIIA(t,\theta )$. The APM problem is formally defined as follows.

\noindent \textbf{Acceptance Probability Maximization (APM). } Given a
social network $G(V,E)$, an initiator $s$ and a friending target $t$, select
a set $R$ of $r_{R}$ users for $s$ to send friending invitations such that
the acceptance probability \\ $ap(t,S,R,MIIA(v,\theta ))$ is maximized, where $%
S $ is the friends of $s$, including $s$ itself.

As analyzed later, the optimal solution to APM can be obtained in $O(n_{V}{%
r_{R}}^{2})$ time\footnote{%
MIA was proposed to simplify IC model, which is computation intensive and
not scalable. Nevertheless, we prove that APM in IC\ model is NP-hard in Theorem~\ref{thm:sharp_P} and not submodular by displaying a counter example in Appendix.}, where $n_{V}$ is
the number of nodes in a social network, and ${r_{R}}$ is the total number
of invitations allowed. The setting of $r_{R}$ has been
discussed in Section \ref{sec:introduction}. It is worth noting that APM
maximizes the acceptance probability of $t$, instead of minimizing the
number of iterations to approach $t$, which can be achieved by the shortest
path routing in an on-line social network. Nevertheless, it is possible to
extend APM by limiting the number of edges in an MIP of $MIIA(t,\theta )$,
to avoid incurring an unaccepted number of iterations in active friending.

\section{Related Work}

\label{sec:related_work}

Recommendation for passive friending has been explored in the past few
years. Chen et al.~\cite{Chen09CHI} manifest that friending recommendations
based on the topology of an on-line social network are the easiest way
leading to the acceptance of an invitation. In contrast, recommendations
based on contents posted by users are very powerful for discovering
potential new friends with similar interests \cite{Chen09CHI}. Meanwhile,
research shows that preference extracted from social networking applications
can be exploited for recommendations \cite{Guy09RecSys}. To avoid
recommending socially distant candidates, users are allowed to specify
different social constraints \cite{Ronen09EDBT}, e.g., the distance between
a user and the recommended friending targets, to limit the scope of
friending recommendation. Moreover, community information has been explored
for recommendation \cite{Spertus05KDD}. Notice that the aforementioned
research work and ideas are proposed for passive friending, where the
friending targets are determined by the recommendation engines of social
networking service providers in accordance with various criteria (e.g.,
preference and social closeness, etc). Thus, the user can conveniently (but
\emph{passively}) send an invitation to targets on the recommendation list.
Complementary to the conventional passive friending paradigm, in this paper,
we propose the notion of active friending where a friending target can be
specified by the initiator. Accordingly, the recommendation service may
assist and guide the initiator to \emph{actively} approach a target.

The impact of social influence has been demonstrated in various
applications, such as viral marketing \cite%
{Chen10KDD,Kempe03KDD,Leskovec07KDD} and interest inference \cite{Wen10KDD}. Given
an on-line social network, a major research problem is the seed selection
problem, where the seeds correspond to the leaf nodes of MIA (i.e.,
initiator $s$ and her friends) in our problem. In contrast, APM is to select the
topology between the friends and $t$, instead of selecting the seeds. The
homophily factor, capturing the tendency of users to connect with similar
ones, has been considered in several applications, such that identifying
trusted users \cite{Tang13WSDM} and users relationships \cite{Yang12SIGIR}
in social networks.

Notice that some works develop algorithms to return a subgraph or path, such
as community detection \cite{Macropol10VLDB}, shortest path \cite%
{Cohen02SODA}, pattern matching \cite{Fan10VLDB}, or graph isomorphism query
\cite{Cheng07SIGMOD}. In contrast to the shortest path query, our algorithms
for the APM problem make emphasize on returning a graph, instead of a path. The
topology of the returned graph contains valuable neighborhood information of
some common friends who can be leveraged to effectively increase the
acceptance probability of a friending invitation. The initiator of a pattern
matching or a graph isomorphism query needs to specify a subgraph as the
query input. In contrast, this paper aims at finding an unknown graph
between $s$ and $t$ to maximize the acceptance probability of a invitation
to a friending target. 

\section{Algorithm Design}

\label{sec:algorithm}

To tackle the APM problem, we aim to design efficient algorithms in support
of the invitation recommendations for active friending. From our earlier discussions
, it is easy to observe that the set of intermediate nodes in $R$,
i.e., those to be recommended for invitation, play a crucial role in
maximizing the acceptance probability for active friending. Here we first
introduce a range-based greedy algorithm which provides some good
insights for our other algorithms.

The algorithm, given an invitation budget $r_{R}$, aims to find the set of
invitation candidates for recommendations to an initiator $s$ who would like
to make friends with a target $t$. 
Let $R$ denote the answer set, which is initialized as empty at the
beginning. The algorithm iteratively selects a node $v$ from the neighbors of $s$'s
current friends and adds it to $R$ based on two heuristics: 1) the
highest acceptance probability and 2) the number of remaining invitations.
The former aims to minimize the potential waste of a friending invitation,
while the latter avoids selecting a node too far away to reach $t$ by
constraining that $v$ can only be at most $r_{R}-\left\vert R\right\vert -1$
hops away from $t$. As a result, the range-based greedy algorithm is
inclined to first expand the friend territory of $s$ and then approach
towards the neighborhood of $t$ aggressively.

\subsection{Selective Invitation with Tree Aggregation}

\label{subsec:single_target}


While the range-based greedy algorithm is intuitive, the nodes added to $R$
at separate iterations are not selected in a coordinated fashion. Thus, it
is difficult for the range-based greedy algorithm to effectively maximize
the acceptance probability. To address this issue, we propose a dynamic
programming algorithm, call \textit{Selective Invitation with Tree
Aggregation (SITA)}, that finds the optimal solution for APM by exploring
the maximum influence in-arborescence tree rooted at $t$ (i.e., $%
MIIA(t,\theta)$) in a bottom-up fashion. SITA starts from the leaf nodes,
i.e., nodes without in-neighbors, to explore $MIIA(t,\theta )$ in a
topological order until $t$ is reached finally. In order to obtain the
optimal solution, SITA needs to explore various allocations of the $r_{R}$
invitations to different nodes close to $s$ or $t$ in $MIIA(t,\theta ) $.
However, it is not necessary for SITA to enumerate all possible invitation
allocations. Thanks to the tree structure of $MIIA(t,\theta )$, for each
node $v$, SITA systematically summarizes the best allocation for $v$, i.e.
which generates the highest acceptance probability for $v$, corresponding to
the subtree rooted at $v$. 
The summaries will be exploited later by $v$'s parent node, i.e., the only
out-neighbor of $v$, to identify the allocation generating the highest
probability. The above procedure is repeated iteratively until $t$ is
processed, and the allocation of $r_{R}$ invitations to the subtree rooted
at $t$ is the solution returned by SITA.

More specifically, let $f_{v,r}$ denote the maximum acceptance probability
for $v$ to accept the invitation from $s$ while $r$ invitations have been
sent to the subtree rooted at $v$ in $MIIA(t,\theta )$. By first sorting all
nodes in topological order to $t$, we process $f_{v,r}$ of a node $v$ after
all $f_{u,r}$ of its in-neighbors $u$ have been processed. Apparently, $%
f_{v,0}=0$ for every node $v$ that is not a friend of $s$ because no
invitation will be sent to the subtree rooted at $v$. In contrast, for every
leaf node $v$, which is a friend of $s$ (or $s$ itself), $f_{v,r}=1$ for $%
r=0 $. 
For all other nodes $v$ in $MIIA(t,\theta )$, SITA derives $f_{v,r}$
according to each in-neighbor $f_{u_{i},r_{i}}$ as follows,
\vspace{-5pt}
\begin{equation}
f_{v,r}=\max_{\sum {r_{i}}=r-1}\{1-\prod_{u_{i}\in N^{in}(v)}{%
[1-f_{u_{i},r_{i}}\cdot w_{u_{i},v}]}\},  \label{eq:f}
\end{equation}%
where $N^{in}(v)$ denotes the set of in-neighbors of $v$ with \\
$\left\vert N^{in}(v)\right\vert =d_{v}$, $u_{i}$ is an in-neighbor of $v$,
and $r_{i}$ is the number of invitations sent from $s$ to the subtree rooted
at $u_{i}$. An invitation is sent to $v$, while the remaining $r-1$
invitations are distributed to the in-neighbors of $v$. SITA effectively
avoids examining all possible distribution of the $r-1$ invitations to the
nodes in the subtree. Instead, Eq. (\ref{eq:f}) examines only $%
f_{u_{i},r_{i}}$ of each in-neighbor $u_{i}$ of $v$ on every possible number
of invitations $r_{i}$. In other words, only the in-neighbors of $v$,
instead of all nodes in the subtree, participate in the computation of $%
f_{v,r}$ to efficiently reduce the computation involved. For each node $v$, $%
f_{v,r}$ is derived in ascending order of $r$ until reaching $r=min({r_{R}}%
,z_{v})$, where $z_{v}$ is the number of nodes that are not friends of $s$
in the subtree rooted at $v$. SITA stops after $f_{t,r_{R}}$ is obtained. In
the following, we show that SITA finds the optimal solution to APM.\footnote{%
Due to the space constraint, we do not show the pseudo-code of SITA\ here
but refer the readers to the next section where a more general SITINA is
presented.}


\begin{lemma}
Algorithm SITA answers the optimal solution to APM. \label{lamma:f}
\end{lemma}

\begin{proof}
We prove the lemma by contradiction. Assume that the solution from SITA,
i.e., $f_{t,r_{R}}$, is not optimal. According to the recurrence, there must
exist at least one in-neighbor $t_{1}\in N_{t}^{in}$ together with the
number of invitations $r_{1}$ such that $f_{t_{1},r_{1}}$ is not optimal.
Similarly, since $f_{t_{1},r_{1}}$ is not optimal, there exists at least one
in-neighbor $t_{2}\in N_{t_{1}}^{in}$ of $t_{1}$ with the number of
invitations $r_{2}$ such that $f_{t_{2},r_{2}}$ is non-optimal, $r_{2}<r_{1}$%
. Here in the proof, let $f_{t_{i},r_{i}}$ denote the non-optimal solution
found in $i$-th iteration of the above backtracking process, which will
continue and eventually end with a probability $f_{t_{i},r_{i}}$ such that
1) $t_{i}$ is a friend of $s$ but $f_{t_{i},0}\neq 1$ or $%
f_{t_{i},r_{i}} = 0$ for $r_{i}>0$, or 2) $t_{i}$ is not a friend of $s$
but $f_{t_{i},0}\neq 0$. The above two cases contradict the initial
assignment of SITA. The lemma follows.
\end{proof}

\noindent \textbf{Example. }
\begin{figure}[t]
\includegraphics[width=3.2 in]{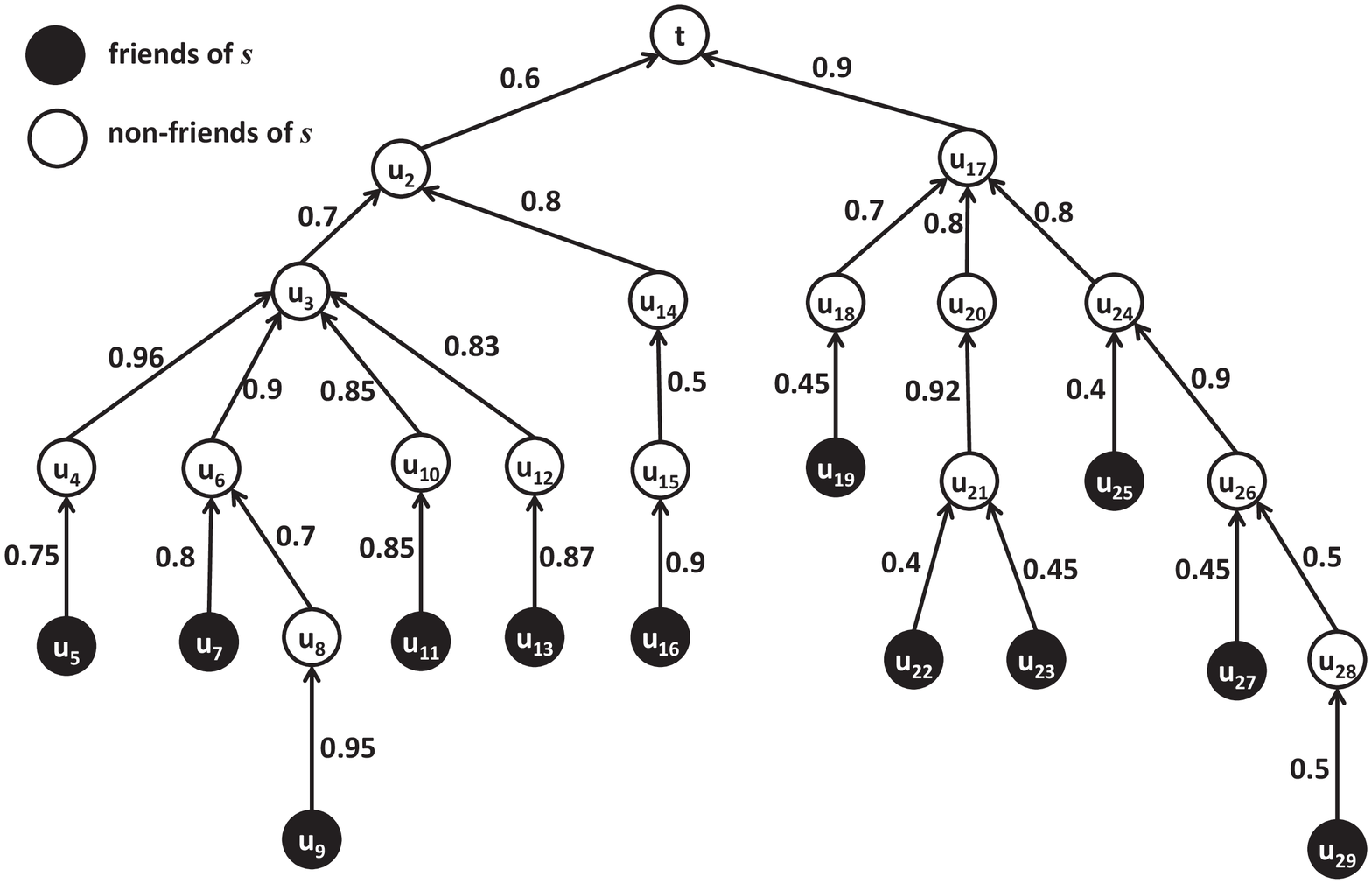} \vspace{-10pt}
\caption{The running example(not including $s$ and her edges)}
\label{fig:example}
\end{figure}
Figure~\ref{fig:example} illustrates an example of $MIIA(t,\theta )$ with $%
r_{R}=7$, where the nodes denote the users involved in deriving the maximal
acceptance probability for $t$ and the numbers labeled on edges denote the
influence probability between two nodes. Without loss of generality, $s$ and her homophily edges are not shown in Figure~\ref{fig:example}.
Note that the dark nodes at the leaf are $s$ and her existing friends and thus have the acceptance
probability as 1, while the white nodes are the recommendation candidates to
be returned by SITA along with their acceptance probabilities. SITA explores
$MIIA(t,\theta )$ from the dark leaf nodes in a topological order, i.e., the
$f_{v,r}$ of a node $v$ is derived after all $f_{u,r}$ of its in-neighbors $%
u $ are processed. Take $u_{4}$ as an example. $f_{u_{4},0}=0$ since no
invitation is sent and $f_{u_{4},1}=1-(1-f_{u_{5},0}\cdot 0.75)=0.75$.
Similarly, for $u_{8}$, $f_{u_{8},0}=0$ and $f_{u_{8},1}=0.95$. Consider $%
u_{6}$ which has in-neighbors $u_{7}$ and $u_{8}$, $f_{u_{6},0}=0$, $%
f_{u_{6},1}=0.8$, and $f_{u_{6},2}=1-(1-f_{u_{7},0}\cdot
0.8)(1-f_{u_{8},1}\cdot 0.7)=0.933$. Notice that for a node $v$, $f_{v,r}$
is derived for $r\in \lbrack 0,min(z_{v},{r_{R}})]$, e.g., for $u_{6}$, we
only derive $r\in \lbrack 0,2]$. Nevertheless, to find $f_{v,r}$, SITA needs
to try different allocations by distributing the number of invitations $%
r_{i} $ to each different neighbor $u_{i}$ and then combining the solutions $%
f_{u_{i},r_{i}}$ to acquire $f_{v,r}$. For example, to derive $f_{u_{17},5}$%
, it is necessary to distribute 4 invitation to its in-neighbors, including $%
u_{18}$, $u_{20}$ and $u_{24}$. The possible allocations for $%
(r_{18},r_{20},r_{24})$ include $(0,1,3)$, $(0,2,2)$, $(1,0,3)$, $(1,1,2)$
and $(1,2,1)$\footnote{%
Some allocations are eliminated since $r_{i}\notin \lbrack 0,min({r_{R}}%
,z_{u_{i}})]$.}, which will obtain acceptance probability 0.5738, 0.7539,
0.7081, 0.6674 and 0.7639 respectively.
Eventually, we obtain $f_{u_{17},5}=0.7639$. Notice that the number of
possible allocations grows exponentially. After all the nodes are processed,
we obtain $f_{t,7}=0.7483$. On the other hand, the greedy algorithm RG
selects a user $v\notin R$ with the highest acceptance probability and at
most $r_{R}-|R|-1$ hops away from $t$. Accordingly, it selects $u_{8}$, $%
u_{6}$, $u_{15}$, $u_{12}$, and $u_{3}$ sequentially. In the 6th step, the node
with the highest probability is $u_{10}$. However, $u_{10}$ is 3-hops away
with $3>r_{R}-|R|-1=2$ and thus not selected. Instead, it selects the node
with the next highest acceptance probability, i.e., $u_{2}$. In the last
step, only the root $t$ can be selected, so RG obtains a solution with the
acceptance probability as 0.4013. As shown, SITA outperforms RG.

%

\subsection{Selective Invitation with Tree and In-Node Aggregation}

\label{sec:SITINA}

Unfortunately, SITA is not a polynomial-time algorithm because in Eq. (\ref%
{eq:f}), $O(r^{d_{v}})$ allocations are examined to distribute $r_{i}-1$
invitations to the subtrees of the $d_{v}$ in-neighbors corresponding to
each node $v$. To remedy this scalability issue, we propose \textit{%
Selective Invitation with Tree and In-Node Aggregation (SITINA)} to answer
APM in polynomial time. SITINA\ effectively avoids processing of $%
O(r^{d_{v}})$ allocations by iteratively finding the best allocation for the
first $k$ in-neighbors, which in turn is then exploited to identify the best
allocation for the first $k+1$ in-neighbors. The process iterates from $k=1$
till $k=d_{v}$. Consequently, the possible allocations for distributing $%
r_{i}-1$ invitations to all in-neighbors are returned by Eq. (\ref{eq:f}) in
$O(d_{v}r_{R})$ time, where $d_{v}$ is the in-degree of $v$ in $%
MIIA(t,\theta )$.

To efficiently derive $f_{v,r}$ in Eq. (\ref{eq:f}), we number the
in-neighbors of $v$ as $u_{1}$, $u_{2}$... to $u_{d_{v}}$, where $d_{v}$ is
the in-degree of $v$. Let $m_{v,k,x}$ denote the maximum acceptance
probability by sending $x$ invitations to the subtrees of the first $k$
neighbors of $v$, i.e., $u_{1}$ to $u_{k}$. Initially, $%
m_{v,1,x}=f_{u_{1},x} $, $x\in \lbrack 0,r_{R}]$. SITINA derives $m_{v,k,x}$
according to the best result of the first $k-1$ in-neighbors,
\begin{equation}
m_{v,k,x}=\max_{x^{\prime }\in \lbrack
0,min(z_{u_{k}},x)]}\{1-[1-m_{v,k-1,x-x^{\prime }}][1-f_{u_{k},x^{\prime
}}w_{u_{k},v}]\},  \label{eq:m}
\end{equation}
where $f_{u_{k},x^{\prime }}w_{u_{k},v}$ is the acceptance probability for
allocating $x^{\prime }$ invitations to the $k$-th in-neighbor $u_{k}$, and $%
m_{v,k-1,x-x^{\prime }}$ is the best solution for allocating $x-x^{\prime }$
invitations to the first $k-1$ in-neighbors. By carefully examining
different $x^{\prime }$, we can obtain the best solution $m_{v,k,x}$ for a
given $k$.

SITINA starts from $k=1$ to $k=d_{v}$. For each $k$, SITINA begins with $x=0$
until $k=\min (\sum_{i\in \lbrack 1,k]}{z_{u_{i}},r_{R}}-1)$, where $%
\sum_{i\in \lbrack 1,k]}{z_{u_{i}}}$ is the total number of nodes that are
not friends of $s$ in the subtrees of the first $k$ in-neighbors. SITINA
stops after finding every $m_{v,d_{v},x}$, $x\in \lbrack 0,min(z_{v},{r_{R}}%
-1)]$. The pseudocode is presented in Algorithm~\ref{alg:base}, and the
following lemma indicates that the optimal solution of APM is $m_{t,d_{t},{%
r_{R}}-1}$.


\begin{lemma}
For any $v$ and $r$, $f_{v,r}=m_{v,d_v,r-1}$. \label{lemma:feqm}
\end{lemma}

\begin{proof}
We prove the lemma by contradiction. Assume that $m_{v,d_{v},r-1}$ is not
optimal. According to the recurrence, there exists at least one $r_{d_{v}}$
such that $m_{v,d_{v}-1,(r_{R}-1-r_{d_{v}})}$ is not optimal. Similarly,
since $m_{v,d_{v}-1,(r-1-r_{d_{v}})}$ is not optimal, there exists at least
one $r_{d_{v}-1}$ such that \\ $m_{v,d_{v}-2,(r-1-r_{d_{v}}-r_{d_{v}-1})}$ is not optimal. Therefore, let \\
$m_{v,d_{v}-i,(r-1-\sigma _{i})}$, where $\sigma _{i}=\sum_{j\in \lbrack
0,i-1]}{r_{d_{v}-j}}$, denote the non-optimal solution obtained in the $i$%
-th iteration. The backtracking process continues and eventually ends with $%
i=d_{v}-1$, where $m_{v,1,r_{1}}\neq f_{u_{1},x}$. It contradicts the
initial assignment of $m_{v,1,r_{1}}$, and the lemma follows.
\end{proof}

The following theorem proves that the algorithm answers the optimal solution
to APM in $O(n_{V}{r_{R}}^{2})$ time, where $n_{V}$ is the number of nodes
in a social network, and $r_{R}$ is the number of invitations in APM. Note that
any algorithm for APM is $\Omega (n_{V})$ time because reading $%
MIIA(t,\theta )$ as the input graph requires $\Omega (n_{V})$ time.
Therefore, SITINA is very efficient, especially in a large social network
with $n_{V}$ significantly larger than $d_{max}$ and $r_{R}$.


\begin{theorem}
SITINA Algorithm answers the optimal solution to APM in $O(n_{V}{r_{R}}^{2})$
time.
\end{theorem}

\begin{proof}
According to Lemma~\ref{lamma:f} and Lemma~\ref{lemma:feqm}, SITINA obtains
the optimal solution of APM. Recall that $n_{V}$ is the number of nodes in
the social network, and $d_{v}$ is the in-degree of a node $v$ in $MIIA(t,\theta )$. The algorithm contains $O(n_{V})$ iterations. Each
iteration examines a node $v$ to find $m_{v,d_{v},x}$ for every $x\in
[0,\min (z_{v}-1,r_{R}-1)]$, where $r_{R}$ is number of invitations sent by $s$ in APM.
There are $O(d_{v}r_{R})$ cases to be considered to explore all $%
m_{v,d_{v},x}$ for $v$ in Eq. (\ref{eq:m}), and each case requires $O(r_{R})$
time. Therefore, finding $m_{v,d_{v},x}$ for a node $v$ needs $O(d_{v}{r_{R}}^{2})$ time, and for all nodes in $MIIA(t,\theta)$ is $O(\sum_{v}{d_{v}{r_{R}}^{2}})$, where $\sum_{v}{d_{v}}=|E|$. As $MIIA(t,\theta)$ is a tree (i.e. $|E|=n_V-1$), the overall time complexity is $O(n_{V}{r_{R}}^{2})$.
The theorem follows.
\end{proof}

\noindent \textbf{Example. }
\begin{table}[tbp]
\begin{tabular}{|c|c|c|c|c|c|c|}
\hline
$k$ & $x=1$ & $x=2$ & $x=3$ & $x=4$ & $x=5$ & $x=6$   \\ \hline
$1$ & 0.315 & * & * & * & * & *   \\
$2$ & 0.315 & 0.4931 & 0.6528 & * & * & *   \\
$3$ & 0.32 & 0.5342 & 0.6674 & 0.7639 & 0.8314 & 0.8520   \\ \hline
\end{tabular}
\label{table:m_example} 
\caption{All $m_{u_{7},k,x}$}
\end{table}
In the following, we illustrate how SITINA derives $f_{u_{17},r}$, $r\in
[0,z_{u_{17}}]$. At the beginning, the in-neighbors of $u_{17}$ are ordered
as $u_{18}$, $u_{20}$ and $u_{24}$.\footnote{%
To avoid confusing, we keep their ID in this example without renaming them
as $u_{u_{17}}^{1}$, $u_{u_{17}}^{2}$, and $u_{u_{17}}^{3}$.} Then, we find
all $m_{u_{17},1,x}=f_{u_{18},x-1}w_{u_{18},u_{17}}$, $x\in \lbrack 0,min({%
z_{u_{18}}},r_{R}-1)]$ first, representing the maximum acceptance
probability $u_{17}$ obtained by only sending $x$ invitations to the subtree
rooted at the first in-neighbor, i.e., $u_{18}$. Then we derive $%
m_{u_{17},2,x}$ for $x\in {[0,min(z_{u_{18}}+z_{u_{20}},r_{R}-1)]}$ to
acquire the maximum acceptance probability of $u_{17}$ by sending
invitations to subtrees rooted at $u_{18}$ and $u_{20}$. Notice that
different $x^{\prime }$, representing the invitations distributed to the $k$%
-th subtree, needs to be examined in order to find the optimal solution. For
instance, while deriving $m_{u_{17},3,4}$, we compare $%
1-(1-m_{u_{17},2,1})(1-f_{u_{24},3}\times w_{u_{24},u_{17}})=0.7081$, $%
1-(1-m_{u_{17},2,2})(1-f_{u_{24},2})\times w_{u_{24},u_{17}})=0.7539$ and $%
1-(1-m_{u_{17},2,3})(1-f_{u_{24},1})\times w_{u_{24},u_{17}})=0.7417$ and
obtain $m_{u_{17},3,4}=0.7539$. After deriving all $%
f_{u_{17},x+1}=m_{u_{17},d_{v},x}$ for $x\in \lbrack 0,6]$ ($%
min(r_{R}-1,z_{u_{17}-1})=6$), the computation of $u_{17}$ finishes. Table 1 lists the detailed results, where * denotes the
instances with $r$ exceeding the number of people who are not the friends of $%
s$ in the first $k$ subtrees.\footnote{%
Note that $m_{u_7,k,x}=0$ when $x=0$.}


\begin{algorithm}[t]
\caption{Selective Invitation with Tree and In-Node Aggregation (SITINA)}
\label{alg:base}
\begin{algorithmic}[1]
\REQUIRE The query issuer $s$; the targeted user $t$; the influence tree $MIIA(t,\theta)$ rooted at $t$; the number of requests $r_R$ that $s$ can send.
\ENSURE A set $R$ of selected users that $s$ sends requests to, such that the acceptance probability is maximized.

\STATE {Obtain a topological order $\sigma$ which orders a node without in-neighbor first.}
\FOR {$v \in \sigma$}
    \STATE{//obtain all $f_{v,r}$, $r \in [0,\min(r_R, n_v)$]}
    \STATE {Order in-neighbors of $v$ as $u^1_v$, $u^2_v$,... $u^{d_v}_v$}
    \STATE {$m_{v, 0, r} \leftarrow 0$ for $\forall r \in [0, min(r_R-1, n_r-1)]$}
    \FOR{ $k = 1$ \textbf{to} $d_v$ }
        \FOR{ $r = 1$ \textbf{to} $min(r_R-1, n_v-1)$}
            \STATE {$x = r-1$}
            \STATE {$m_{v,k,x} = 0$}
            \FOR {$x' = 0$ to $r$}
                \IF{$m_{v,k,x} < 1-[1-m_{v,k-1,x-x'}][1-f_{u^k_v,x'}w_{u^k_v,v}]$}
                    \STATE{$m_{v,k,x} \leftarrow 1-[1-m_{v,k-1,x-x'}][1-f_{u^k_v,x'}w_{u^k_v,v}]$}
                    \STATE{$\pi_{v,k,x} \leftarrow x'$}
                \ENDIF
            \ENDFOR
        \ENDFOR
    \ENDFOR
    \STATE{$f_{v,0}\leftarrow 0$}
    \STATE{$f_{v,x+1} \leftarrow m_{v,k,x}$, $\forall x \in [0, min(r_R-1, n_r-1))]$}
\ENDFOR
\STATE{Backtrack $\pi_{v,k,x}$ to obtain $R$}
\RETURN{$R$ with maximized $f_{t, r_R}$}
\end{algorithmic}
\end{algorithm}


\section{Performance Evaluation}

\label{sec:experiments}

We implement active friending in Facebook and conduct a user study and a comprehensive set of experiments to validate
our idea of active friending and to evaluate the performance of the proposed
algorithms. In the following, we first detail the methodology of our
evaluation and then present the results of our user study and experiments,
respectively.

\subsection{Methodology}

We adopt a user study and experiments, two complementary approaches, for the
performance evaluation. We aim to use the user study to investigate how the
recommendation-based active friending approach is faring with the approach
based on the users' own strategies (i.e., which they would follow under the
existing environment of social networking services). To perform the user
study, we implement an app. on Facebook.
Through the app., the user is able to decide whom to invite based on their
own strategies to approach the target. Meanwhile, according to the
recommendations generated from the Range-Based Greedy (RG) algorithm and the
Selective Invitation with Tree and In-Node Aggregation (SITINA) algorithm,
respectively, the user also sends alternative sets of invitations to proceed
the active friending activities for comparison.\footnote{%
To alleviate the burden of the participants, we send
invitations on their behalves to the recommended candidates.} Note that
Selective Invitation with Tree Aggregation (SITA) is not considered because it
makes exactly the same recommendations as SITINA. We recruited 169
volunteers to
participate in the user study. 
Each volunteer is given 25 targets with varied invitation budgets to work
on. The social distances between the volunteer and the targets are
pre-determined in order to collect results for comparison under controlled
parameter settings.

On the other hand, we conduct experiments by simulation to evaluate the
solution quality\ and efficiency of SITA, SITINA, and RG, implemented in an
HP DL580 server with four Intel Xeon E7-4870 2.4 GHz CPUs and 128 GB RAM.
Two large real datasets, \textit{FacebookData} and \textit{FlickrData} are
used in the experiments. \textit{FacebookData} contains 60,290 users and
1,545,686 friend links crawled from Facebook~\cite{Viswanath09WOSN}, and
\textit{FlickrData} contains 1,846,198 users and 22,613,981 friend links
crawled from Flickr~\cite{Mislove07IMC}. The initiator $s$ and target $t$
are selected uniformly at random.

An important issue faced in both of our user study and the experiments is
the social influence and homophily factors captured in the social network,
which are required for RG, SITA and SITINA to make recommendations.
Most of previous works adopt a fixed probability (e.g., 1/degree in \cite%
{Kempe03KDD,Chen10KDD,Chen11SDM,Chen10ICDM}) or randomly choose a
probability from a set a values (e.g., 0.001, 0.01, 0.1 in \cite%
{Chen10KDD,Chen11SDM}) due to the lack of real social influence
probabilities and homophily probabilities. To address this issue, in the
user study, we obtain the social influence probability on each edge by
mining the interaction history of volunteers in Facebook in accordance with
\cite{Goyal10WSDM,Goyal11VLDB}. We also derive the homophily probabilities
from $s$ to other nodes by mining the profile information in their Facebook
pages based on \cite{Chen09CHI}. The social network in the user study is
denoted as \textit{UserStudyData}. As for the social networks in \textit{%
FacebookData} and \textit{FlickrData} that are to be used for experiments,
we unfortunately do not have personal profiles and historical interactions
of the nodes. Thus, we could not generate the social influence probability
and homophily probability by mining real data. As a result, we choose to
assign the link weights of the social network based on: i) the distributions
of social influence and homophily probabilities obtained from our user study
(denoted as US), and ii) the Zipf distribution for its ability to capture
many phenomena studied in the physical and social sciences~\cite{Zipf1949AWP}%
. 

\subsection{User Study} \label{sec:US}

Through the user study, we have logged the responses of participants to
invitations and thus are able to calculate the acceptance probabilities
corresponding to invitations under various circumstances. Using the
collected data, we make a number of comparisons.

First, we would like to verify that the acceptance probability of an active friending plan derived based on MIIA tree (using the mined social influence and homophily probabilities as the link weights) are consistent with that of the plan being executed in the user study. Towards this goal, we first verify the accuracy of our invitation acceptance model (for single invitation) by comparing the derived acceptance probability and the actual acceptance probability obtained from real activities in the user study.
%
Figure~\ref{fig:user_study_p1}(a), where results obtained from the user study and our model are respectively labeled as {\em Actual} and {\em Derived}, plots the comparison in terms of the number of common friends in an invitation.
As shown, the acceptance probabilities of both User and
Model increase as the number of common friends in invitations increases. Most
importantly, the results are consistently close, showing our invitation acceptance model (and the social
influence and homophily weights used) are able to reasonably capture the decision making upon invitations in real life.


\begin{figure}[t]
\subfigure[Invitation Acceptance] {\includegraphics[ width=1.6 in, height = 0.9
in]{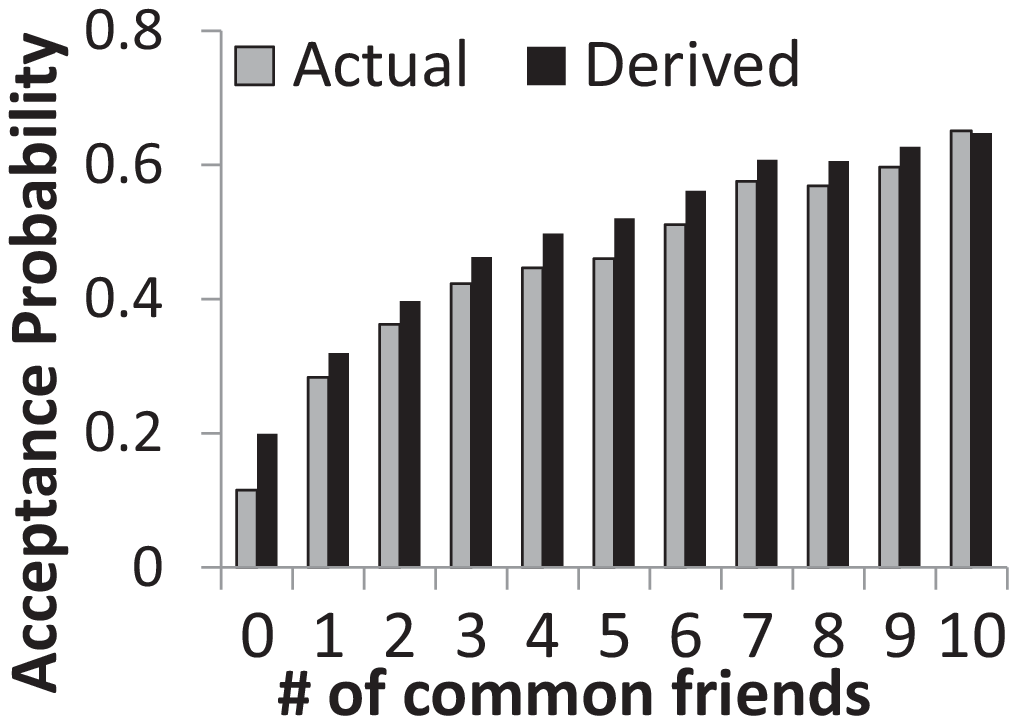}}
\subfigure[MIIA Tree]{\includegraphics[ width=1.6 in, height = 0.9
in]{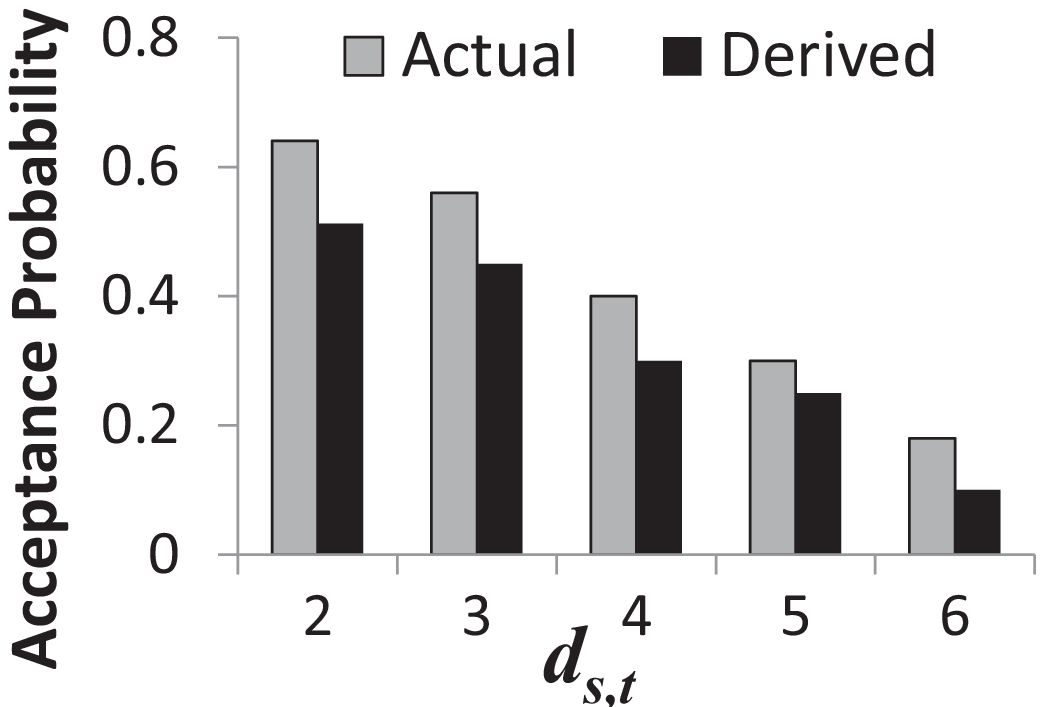}} 
\caption{Verify the derived acceptance probability}
\vspace{5pt}
\label{fig:user_study_p1}
\vspace{5pt}
\end{figure}

Notice that the above comparison focusses on the aspect of invitation acceptance only, without taking into account the social network topology, which we approximate with the MIIA tree.
To verify that using the MIIA tree is sufficiently effective for active friending planning, we further compare the acceptance probability derived using our proposed algorithms and the actual acceptance probability obtained through executing the plan in the user study. Figure~\ref{fig:user_study_p1}(b) shows that, under various distance between initiator and target, the acceptance probabilities derived using MIIA tree is reasonably close to the actual acceptance probabilities.

%


Next, we compare the effectiveness of strategies based on RG, SITINA and the participants' own heuristics. Figure~\ref{fig:user_study_p2}(a) plots the comparison by
varying the number of friending invitations, $r_{R}$.
RG and SITINA generally outperform user heuristics (labeled as User) under all settings. We can observe
that the performance of SITINA is generally very good and getting better as $r_R$ increases, while the performance of User and RG have a leap from $r_{R}=5$ to $10$ and remain close afterwards. This indicates the extra computation effort required for deriving recommendations due to the
increased invitation budget are worthwhile, outperforming the heuristic strategies
derived based on RG and human intuition.
Figure~\ref{fig:user_study_p2}(b) evaluates the acceptance probability of $t$
under varied settings of $d_{s,t}$. When $d_{s,t}$ is 2, it is more likely to have a lot of common friends (due to the nature of social networks) and thus getting better acceptance probabilities. When $d_{s,t}$
increases, it becomes more difficult for an initiator to make effective
decisions due to the less number of common friends and the lack of knowledge about the larger and more complex
social network topology behind.
As shown in Figure~\ref{fig:user_study_p2}, SITINA\ has
the best performance.




\begin{figure}[t]
\subfigure[Different $r_R$($d_{s,t}=3$)] {\includegraphics[ width=1.6 in, height = 0.9
in]{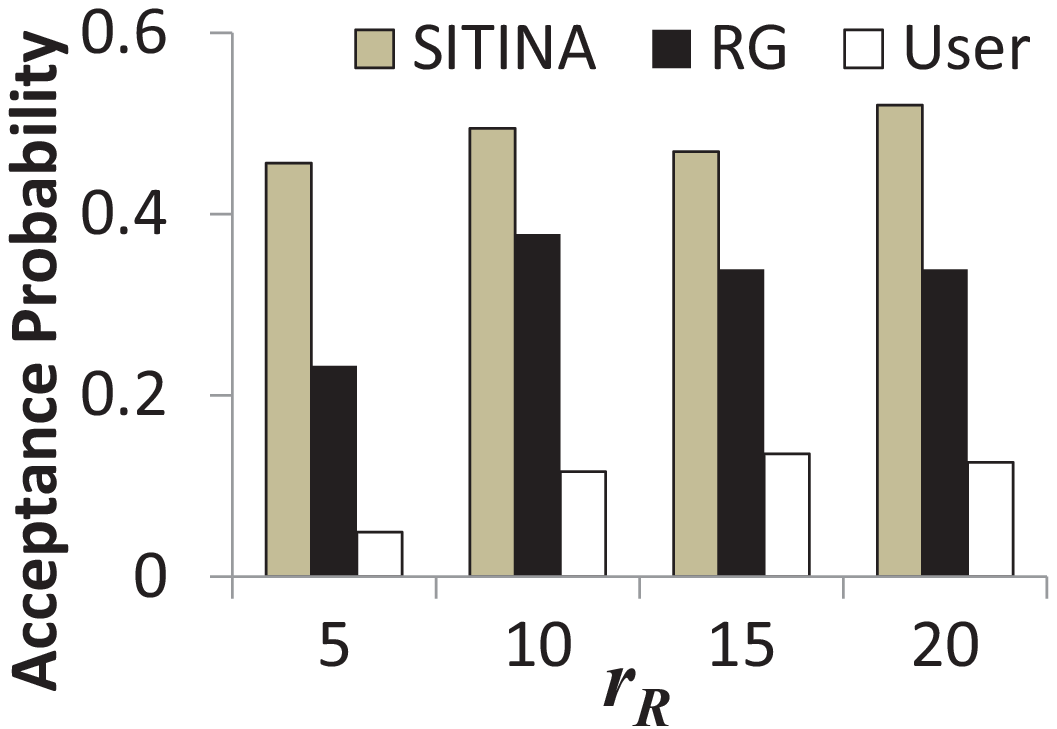}}
\subfigure[Different
$d_{s,t}$($r_R=20$)]{\includegraphics[ width=1.6 in, height = 0.9
in]{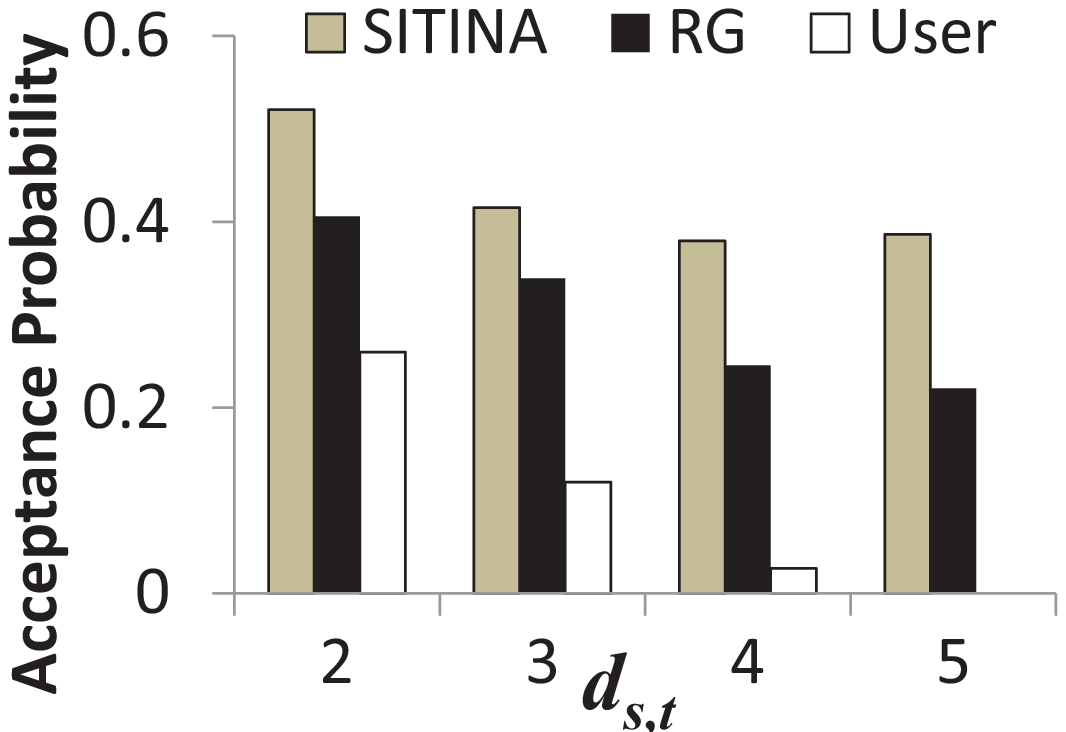}} 
\caption{Acceptance probability in user study}
\label{fig:user_study_p2}
\end{figure}


\subsection{Experimental Results}

While the user study verifies that SITINA is able to
achieve the best performance, the size of social network is small due to the
limited number of volunteers participating in the study. To further validate
our ideas in a large-scale social network and to evaluate the scalability of
SITINA, we conduct an experimental study by simulations.

\subsubsection{Scalability}

\begin{figure}[t]
\begin{minipage}[t]{0.46\linewidth}
\centering
\vspace{4pt}
\includegraphics[width=1.6in, height = 0.9 in]{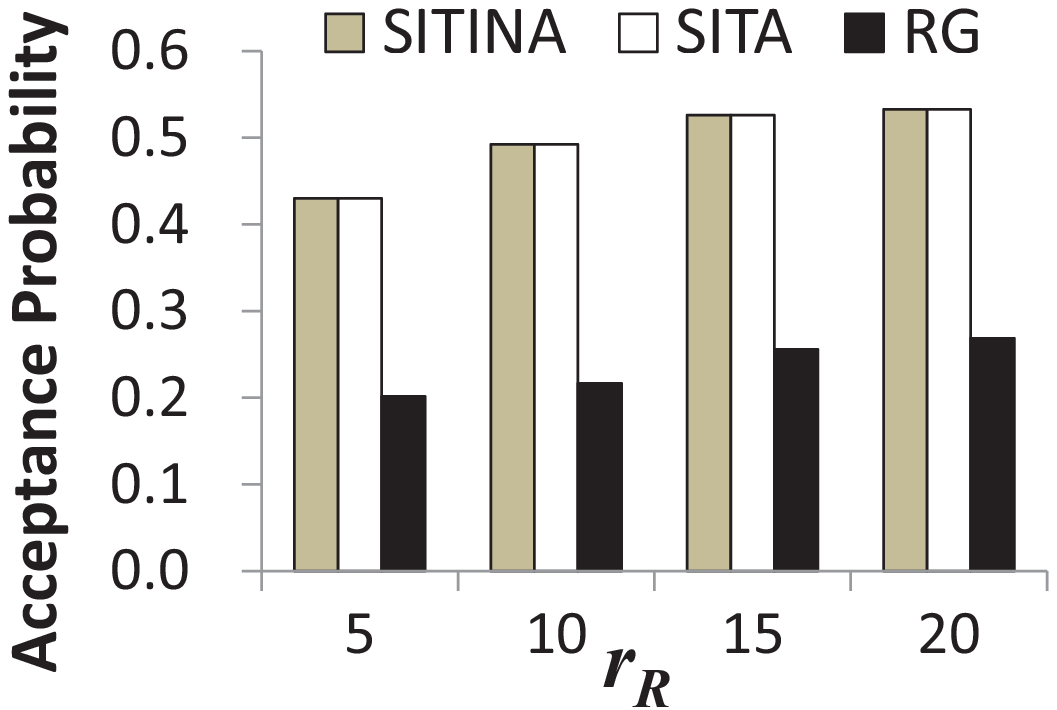}
\caption{Acceptance probability ($d_{s,t}=2$)}\label{fig:SITA}
\end{minipage}
\hspace{2pt}
\begin{minipage}[t]{0.46\linewidth}
\centering
\vspace{8pt}
\includegraphics[width=1.6in]{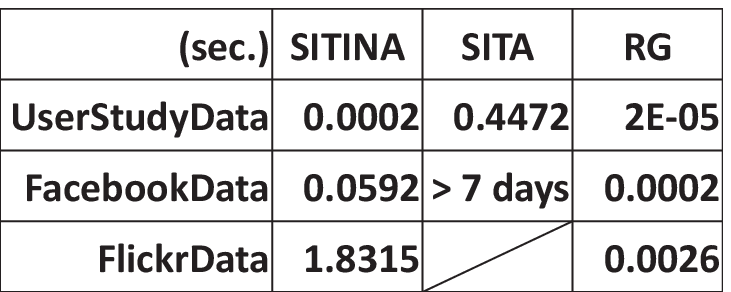}
\vspace{0.10 in}
\caption{Average Running time ($r_R=20$)} \label{fig:running_time}
\end{minipage}
\end{figure}

As proved earlier,
SITA can obtain the optimal
solution of APM. However, it is not scalable as it needs to examine all
combinations of invitation allocations. Here we use it as a baseline to
compare the efficacy and efficiency with SITINA over social networks of
different sizes. First, we compare the results by randomly sampling 50
(initiator, target) pairs using \textit{UserStudyData}.
As Figure \ref{fig:SITA} depicts, both SITA and SITINA significantly
outperform RG in terms of acceptance probability.
Next, we compare their running time, not only using \textit{UserStudyData}
but also the large-scale \textit{FacebookData} and \textit{FlickrData}.
As shown in Figure~\ref{fig:running_time}, the SITA algorithm takes more than 7 days without
returning the answer and thus not feasible for practical use..
For the rest of experiments, we only compare SITINA with RG.


\subsubsection{Sensitivity Tests}

In this section, we conduct a series of sensitivity tests to examine the
impact of different parameters, including the invitation budget ($r_{R}$),
the number of friends of $s$ ($N$), the distance between the initiator and target ($d_{s,t}$), and the skewness of social influence and homophily probabilities ($\alpha$). In experiments on the impacts of $r_R$, $N$, $d_{s,t}$, we have tested both the FacebookData (US)  and FlickrData (US).\footnote{US and ZF denote the link weights assigned based on models
from User Study and Zipfian Distribution.} As the observations on both datasets are quite similar, we only report both results for the first experiment and skip the FlickrData result for the rest due to space constraint. Finally, in the last experiment, we use FacebookData (ZF) to observe how $\alpha$ may potentially impact our algorithms.

%

\noindent \textbf{Impact of $r_{R}$.} By varying $r_{R}$ and setting the
default $d_{s,t}$ of sampled $(s,t)$ pairs as 4, we compare SITINA and RG in
terms of the acceptance probability and the number of iterations using
FacebookData (US) and FlickrData (US) (see Figure~\ref{fig:dif_r_fb} and
Figure~\ref{fig:dif_r_fr}, respectively).
As shown in Figure \ref{fig:dif_r_fb}(a) and \ref{fig:dif_r_fr}(a), SITINA
exhibits much better performance than RG, regardless of the $r_{R}$.
Meanwhile, Figure \ref{fig:dif_r_fb}(b) and \ref{fig:dif_r_fr}(b) manifest
that the longest path in the solution $R$ obtained by SITINA is shorter than
that in RG because RG tends to spend invitations on some local users with
higher acceptance probabilities.\footnote{%
RG is inclined to take more time to reach $t$ because
invitations are sequentially sent towards $t$. The latency of friending a
new intermediate node is different for each node.}

\begin{figure}[t]
\subfigure[Acceptance probability] {\includegraphics[height = 0.9 in,
width=1.6in]{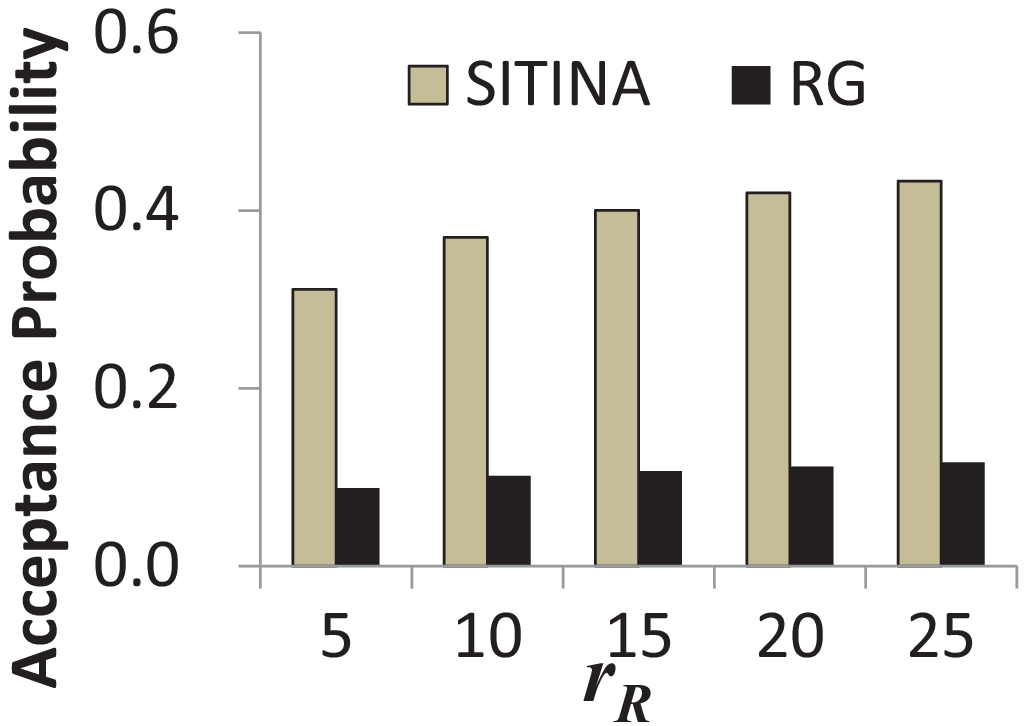}}
\subfigure[Number of edges]{\includegraphics[height = 0.9 in, width=1.6
in]{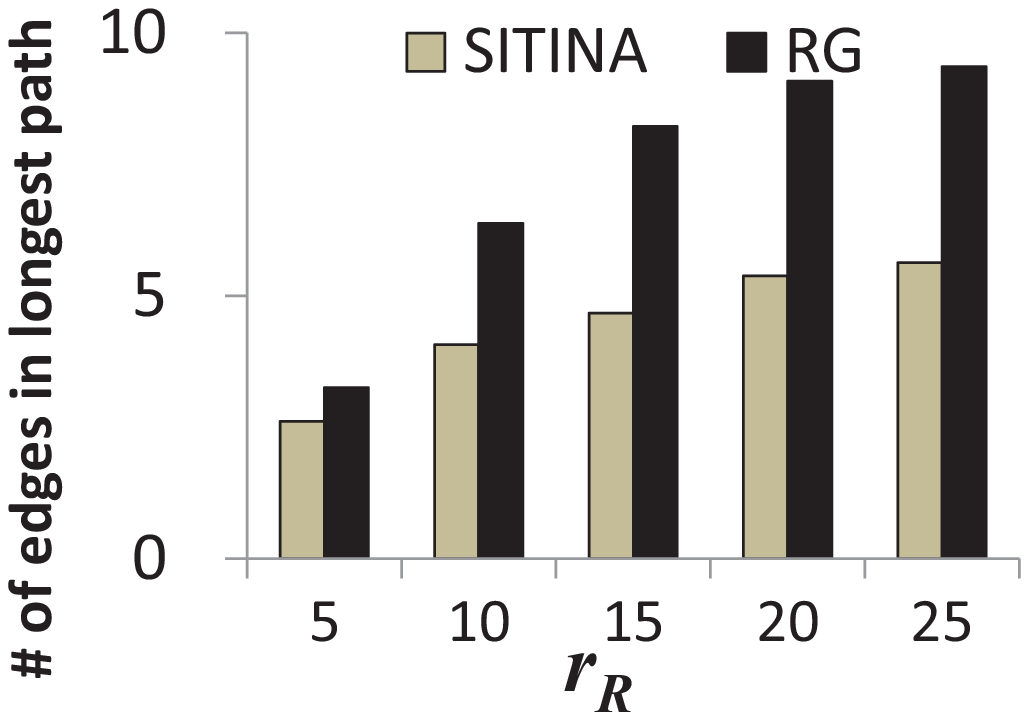}} 
\caption{Varying $r_R$ (\textit{FacebookData} (US))}
\label{fig:dif_r_fb}
\end{figure}
\begin{figure}[t]
\subfigure[Acceptance probability] {\includegraphics[height = 0.9 in,
width=1.6in]{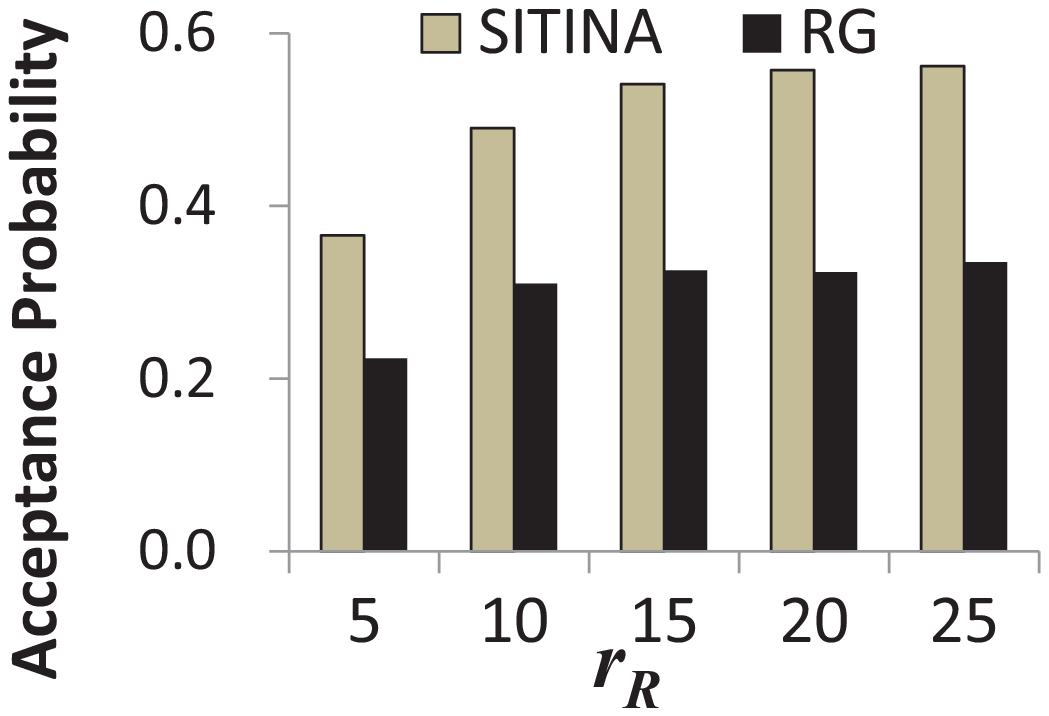}}
\subfigure[Number of edges]{\includegraphics[height = 0.9 in, width=1.6
in]{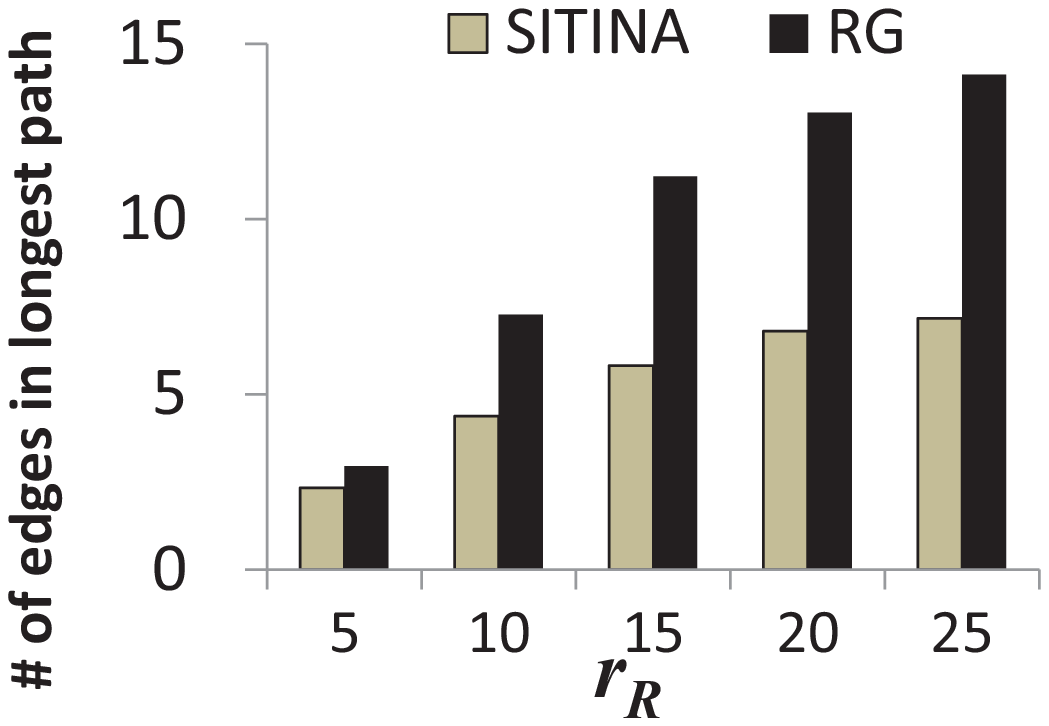}} 
\caption{Varying $r_R$ (\textit{FlickrData} (US))}
\label{fig:dif_r_fr}
\end{figure}

\begin{figure}[t]
\begin{minipage}[t]{0.46\linewidth}
\centering
\vspace{0pt}
\includegraphics[width=1.6in, height = 0.9 in]{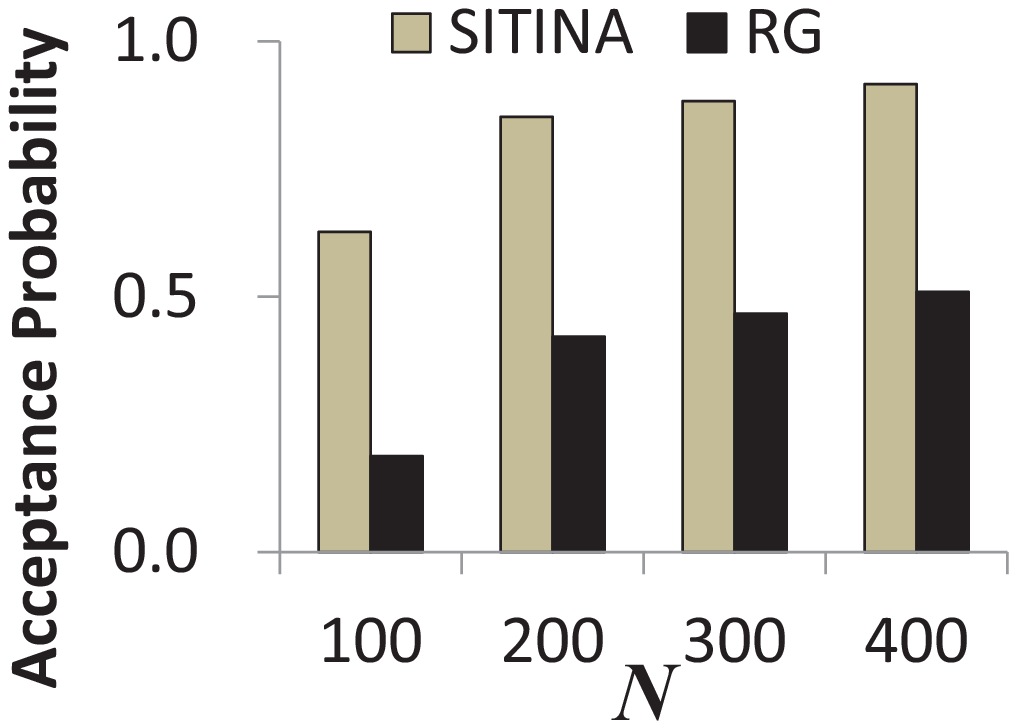}
\caption{Varying $N$ (\textit{FacebookData} (US))}\label{fig:diff_n}
\end{minipage}
\hspace{2pt}
\begin{minipage}[t]{0.46\linewidth}
\centering
\vspace{0pt}
\includegraphics[width=1.6in, height = 0.9 in]{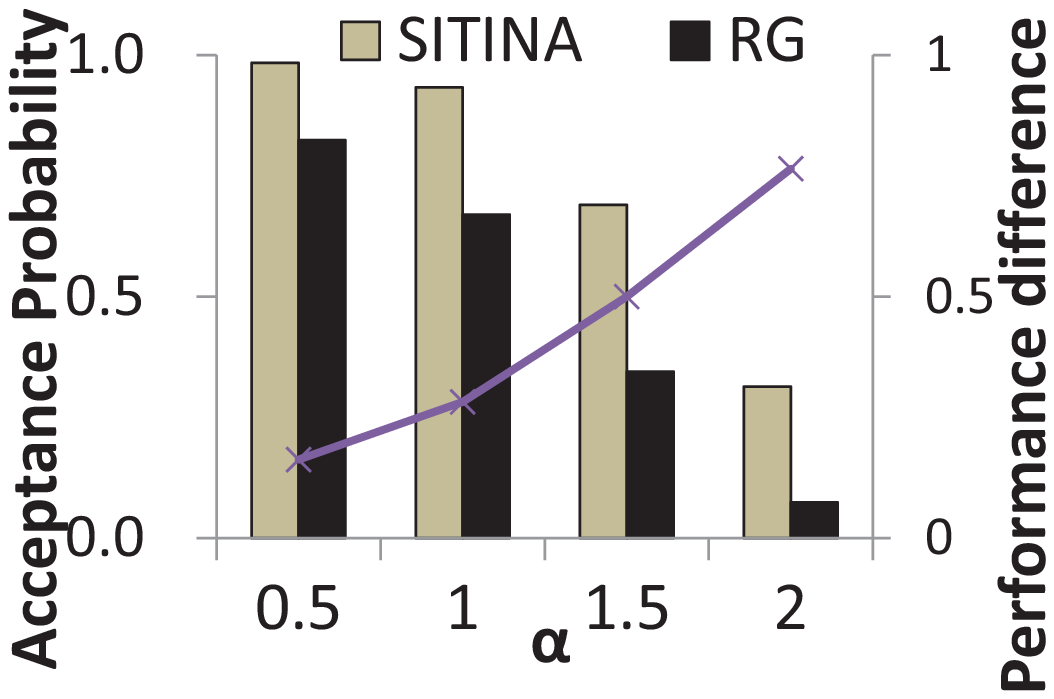}
\caption{Varying $\alpha$ (\textit{FacebookData} (ZF))}\label{fig:ZF}
\end{minipage}
\end{figure}

\noindent \textbf{Impact of $N$.} We are interested in finding whether the number of
friends of $s$ has an impact on the performance. Thus,
we choose four different groups of initiators $s$ (who have around 100, 200, 300, and 400 friends, respectively)
and sample 100 different targets $t$ to compare their acceptance probabilities. With
$r_{R}$ set as 25, Figure \ref{fig:diff_n} shows that as the number of
friends increases, the initiators have more choices to reach their targets.
As shown, SITINA can find the optimal solution with high acceptance probability, while
near-sighted RG tends to select the friends of friends with higher
acceptance probabilities and eventually results in small acceptance
probability to $t$.

\noindent \textbf{Impact of $d_{s,t}$.} We also conduct an experiment to understand the impact of $d_{s,t}$ on the performance. Not surprisingly, the finding is consistent with our user study (please refer to Section~\ref{sec:US} and Figure ~\ref{fig:user_study_p2}(b)). Thus, we do not plot the result here due to the space constraint.







\noindent \textbf{Impact of $\alpha$.} In the experiments above, social influences and homophily factors are modeled based on our user study, but the distributions in different social networks may vary. Thus, through the skewness parameter $\alpha$, we use Zipf distribution, to examine the impact of $\alpha$ on our algorithms. As shown in Figure~\ref{fig:ZF}, we can observe that as the distributions of social influence and homophily become more skewed (i.e., $\alpha$ increases), the acceptance probabilities of SITINA and RG drops, because it becomes more difficult for invitations to get accepted when there are less number of highly influential links while the number of less influential links increases.
It is also worth noting that,as the line in the figure indicates, the percentage of performance difference between SITINA and RG (i.e., acceptance prob. of SITINA divided by that of RG) increases, showing that SITINA is able to handle skewed distribution much better than RG.

\section{Conclusion and Future Work}

\label{sec:conclusion}

Observing the need of active friending in everyday life, this paper
formulates a new optimization problem, named Acceptance Probability
Maximization (APM), for making friending recommendations on-line social
networks. We propose Algorithm Selective Invitation with Tree and In-Node
Aggregation (SITINA), to find the optimal solution for APM and implement
SITINA in Facebook. User study and experimental results manifest that active
friending can effectively maximize the acceptance probability of the
friending target.

In our future work, we will first explore the impact of delay between
sending an invitation and acquiring the result in active friending. This is
important when the user would like to make friends with the target within
certain time frame. In addition, for multiple friending targets, it is not
efficient to configure recommendations separately for each target. An idea
is to give priority to the intermediate nodes that can approach many targets
simultaneously. We will study active friending of a group of targets in the
future work.

\bibliographystyle{abbrv}
\bibliography{reference}

\appendix

We display that APM is not submodular by the following counter example with four users.
The user $a$ is the existing friend of $s$, i.e., $S=\{a\}$, while $b$, $c$, and $t$ are non-friend users of $s$. The influence probability is labeled beside each edge. Consider adding a new user $c$ to two different set of selected users $R_S=\{t\}$ and $R_T=\{b,t\}$, where $R_S \subset R_T$. If the submodular property holds, $ap(t,S,R_S\cup\{c\},MIIA(t,\theta))) - ap(t,S,R_S,MIIA(t,\theta))) \geq  ap(t,S,R_T\cup\{c\},MIIA(t,\theta))) - ap(t,S,R_T,MIIA(t,\theta)))$ should hold.
The acceptance probability of selecting $R_S$, i.e., $ap(t,S,R_S,MIIA(t,\theta)))$, is 0 since there is no path from $a$ to $t$. Similarly, $ap(t,S,R_S\cup\{c\},MIIA(t,\theta)))=0$. The acceptance probability of selecting $R_T$ is $1-(1-0.9\times0.1)=0.09$, and adding $c$ into $R_T$ results in acceptance probability $ap(t,S,R_T\cup\{c\},MIIA(t,\theta))=1-(1-0.09)(1-0.9)=0.909$. However, $ap(t,S,R_S\cup\{c\},MIIA(t,\theta))) - ap(t,S,R_S,MIIA(t,\theta)))=0 < ap(t,S,R_T\cup\{c\},MIIA(t,\theta))) - ap(t,S,R_T,MIIA(t,\theta)))=0.909-0.09=0.819$. There is a counter example and the submodular property does not hold in APM.

\begin{center}
\resizebox{3 in}{0.66 in} {\includegraphics{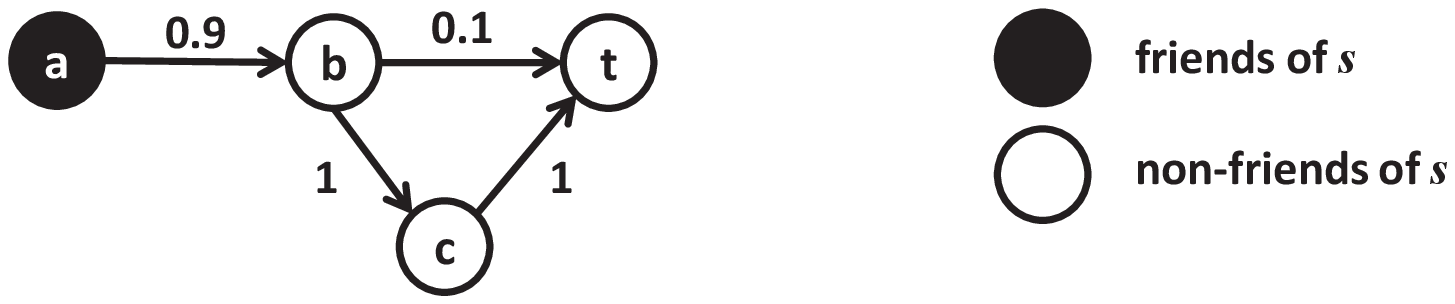}}
\end{center}

\end{document}